\documentclass[twocolumn]{IEEEtran}     
\usepackage{blindtext}

\usepackage{amsmath,epsfig}
\usepackage[T1]{fontenc}
\usepackage[utf8]{inputenc} 
\usepackage{psfrag}
\usepackage{array}
\usepackage{cite}
\usepackage{amsfonts}

\usepackage{amssymb}
\usepackage{color}
\usepackage[dvipsnames]{xcolor}
\usepackage{rotating}
\usepackage{subfigure}
\usepackage{graphicx}
\usepackage{float}
\usepackage{circuitikz}
\usepackage{tikz}
\usepackage{amsbsy} 
\usepackage{tcolorbox}
\usepackage{multirow}

\usepackage{acronym}

\usepackage{algorithm}
\usepackage{algorithmic}

\usepackage{amsthm}

\def\blfootnote{\xdef\@thefnmark{}\@footnotetext}
\makeatother

\makeatletter
\def\blfootnote{\xdef\@thefnmark{}\@footnotetext}
\makeatother

\newtheorem{lemma}{Lemma}
\input{acronym.def}

\title{\huge{Slow Fluid Antenna Multiple Access with Multiport Receivers}}

\begin{document}

\author{Jos\'e P. Gonz\'alez-Coma and F. Javier L\'opez-Mart\'inez,~\IEEEmembership{Senior Member,~IEEE}}

\maketitle

\begin{abstract}
We investigate whether equipping fluid-antenna (FA) receivers with multiple ($L>1$) radiofrequency (RF) chains can improve the performance of the slow fluid-antenna multiple access (FAMA) technique, which enables open-loop connectivity with channel state information (CSI) available only at the receiver side. We analyze the case of slow-FAMA users equipped with multiport receivers, so that $L$ ports of the FA are selected and combined to reduce interference. We show that a joint design of the port selection matrix and the combining vector at each receiver yields significant performance gains over reference schemes, demonstrating the potential of multiport reception in FA systems with a limited number of RF chains.
\end{abstract}


\blfootnote{\noindent Manuscript received XX XX, XXXX. The review of this paper was coordinated by XXXX. This work is supported by grant PID2023-149975OB-I00 (COSTUME) funded by MICIU/AEI/10.13039/501100011033 and FEDER/UE. The authors thank the Defense University Center at the Spanish Naval Academy for their support. This work has been submitted to the IEEE for publication. Copyright may be transferred without notice, after which this version may no longer be accessible}

\blfootnote{\noindent J.P. Gonz\'alez-Coma is with the Defense University Center at the Spanish Naval Academy, 36920 Marín, Spain. Contact email: jose.gcoma@cud.uvigo.es.}
\blfootnote{\noindent F.J. L{\'o}pez-Mart{\'i}nez is with Dept. Signal Theory, Networking and Communications, Research Centre for Information and Communication Technologies (CITIC-UGR), University of Granada, 18071, Granada, (Spain). 
} 

\section{Introduction}

\IEEEPARstart{F}{luid} antenna systems (FASs) have gained considerable attention in recent years as a promising alternative to conventional \ac{MIMO} systems that rely on fixed-position antennas \cite{Wong2021}. FASs refer to any controllable structure —whether electronically, mechanically, or otherwise actuated— that can dynamically alter its fundamental radio-frequency (RF) characteristics, effectively emulating physical changes in shape or position to adapt to the wireless environment \cite{Tutorial2025}. The fine-grained capability of FASs to sample spatial locations brings enhanced flexibility to exploit spatial diversity, by selecting the receive antenna port (akin to position) that is more beneficial for signal reception \cite{New2025}.

One of the key potential use cases of FASs is the provision of simplified open-loop multiple access, without the need for \ac{CSI} at the transmitter side or \ac{SIC} at the receiver side. This technique, referred to as \ac{FAMA}, allows serving multiple users in the same time/frequency resource with CSI required only at the receiver side. Specifically, the slow-\ac{FAMA} paradigm introduced in \cite{SlowFAMA} manages to relax the stringent port-switching requirements of the fast-\ac{FAMA} counterpart \cite{FAMA}, thus reducing complexity while still allowing for user multiplexing.

Since their inception, FASs have been envisioned as a low-complexity alternative to conventional \ac{MIMO} architectures to improve system performance. For that reason, virtually all FAMA-based schemes consider a single \ac{RF} chain at the receiver (user) ends, and the port selection mechanism aims to choose that with the best \ac{SINR}. Indeed, there have been attempts to consider the case with a larger number of RF chains at the receiver ends in the context of FASs: for instance, the use of \ac{MRC} to combine the signals of the $L$ best ports was proposed in \cite{Lai2024} for a point-to-point (single user) case. However, \ac{MRC} offers limited benefits in the presence of interference, which is the operational regime in FAMA. In \cite{CUMA}, the \ac{CUMA} architecture was proposed, allowing to select a subset of ports for which the desired signal is phase-aligned. While \ac{CUMA} manages to improve slow-\ac{FAMA} multiplexing capabilities by using a limited number of RF chains (up to four \cite{CUMA2}), it comes in hand with some practical challenges: (\textit{i}) requires a pre-processing stage to {configure} the subset of active ports; (\textit{ii}) assumes that an arbitrary number of ports are active at the same time; and (\textit{iii}) requires precise knowledge of \ac{CSI} at all ports. 

In this work, we revisit the problem of slow-FAMA when FA users have multiport receivers, i.e., they can be equipped with more than one RF chain. Contrary to prior art, we propose a method to jointly select the active ports of the FA and the combining vector for the signals of the designated ports. One of the key novelties lies in the fact that the spatial features of the channel matrix for the FA array are considered in the procedure. This novel approach enhances the achievable gain and, more importantly, considerably reduces the inter-user interference. Results show that the use of only a few RF chains at the FA users can significantly boost the \ac{SINR} compared to state-of-the-art alternatives.

\textit{Notation:} $a$ is a scalar, $\mathbf{a}$ is a vector, and $\mathbf{A}$ is a matrix. Transpose and conjugate transpose of $\mathbf{A}$ are denoted by $\mathbf{A}^{\operatorname{T}}$ and $\mathbf{A}^{\operatorname{H}}$, respectively. Calligraphic
letters, e.g., $\mathcal{A}$ denote sets and sequences. $\mid \mathcal{A} \mid$ represents the set cardinality. 
Finally, the expectation operator is $\mathbb{E}\left\{\cdot\right\}$ and $\parallel \cdot  \parallel_p$ denotes the $p$-norm.
\section{System Model}
\label{sec:SM}
Let us consider a general \ac{MIMO} system with a \ac{BS} deploying $M$ antennas that serves $K$ users, each equipped with a FA array of $N$ ports capable of selecting $L$ active FA ports based on channel conditions. In coherence with the original formulation of slow-\ac{FAMA} \cite{SlowFAMA}, we assume that the number of antennas at the \ac{BS} $M$ and the number of users $K$ are equal, i.e., $M=K$. 

We denote the data symbol intended for user $k$ as $z_k\in\mathbb{C}$ with $\mathbb{E}\left\{|z_k|^2\right\}=\sigma_S^2$ , so that the received signal at the $k$-th user can be written as:
\begin{equation}
    \mathbf{x}_k = \mathbf{H}_k\mathbf{p}_kz_k +\sum_{j\neq k}\mathbf{H}_k\mathbf{p}_jz_j+ \mathbf{n}_k,
\end{equation}
where $\mathbf{H}_k \in \mathbb{C}^{N \times M}$ is the channel matrix between the base station and user $k$, and $\mathbf{x}_k \in \mathbb{C}^{N \times 1}$ denotes the received signal vector at the FA-user side\footnote{Note that while the vector form of $\mathbf{x}_k$ suggests a 1D implementation of the FA array, the 2D case is also captured by a proper remapping of the 2D FA indices.}. As one of the key features of slow-\ac{FAMA} is its simplicity, each of the transmit antennas is dedicated to a different user, i.e., no channel knowledge is required at the \ac{BS} \cite{SlowFAMA}. Accordingly, the precoding vector at the BS is simply set as $\mathbf{p}_k=\mathbf e_k$, where $\mathbf e_k$ is the $M$-dimension canonical vector with zero entries except for the $k$-th position. We consider that such $\mathbf{p}_k$ is given and fixed. Finally, $\mathbf{n}_k \in \mathbb{C}^{N \times 1}$ represents the additive Gaussian white noise vector whose elements have $\sigma^2$ power.

At the receiver side, each FA-equipped user selects the appropriate $L$ ports based on the user channel gain and interference. We model this feature by introducing the port selection matrix $\mathbf{S}_k\in\mathcal{B}$, with the set of matrices $\mathcal{B} := \left\{ \mathbf Z \in \{0,1\}^{N \times L}, \|\mathbf Z\|_{0,\infty} \leq 1 \right\}$. Moreover, the resulting signals from the $L$ ports can be coherently received by using the combiner $\mathbf{w}_k\in\mathbb{C}^L$, with $\|\mathbf w_k\|_2 = 1$, $\forall k$. As such, the received symbol results\footnote{Note that the matrix $\mathbf{S}_k$ and the vector $\mathbf{w}_k$ can be designed independently (locally) for each user; since they are implemented at the receiver ends, they do not affect the amount of interference suffered by the remaining users in the system.} in
\begin{equation}
    \hat{z}_k=\mathbf{w}_k^H\mathbf{S}_k^T\mathbf{x}_k.
\end{equation}
From a channel modeling perspective, the fluid antenna array resembles a collection of co-located radiating elements, representing the ports where the fluid antenna can switch into. Thus, when considering the case of spatially correlated Rayleigh fading, the $m$-th column of the channel matrix $\mathbf H_k$ can be expressed as 
\begin{equation}
    \mathbf{h}_{k,m} \sim \mathcal{N}_\mathbb{C}(\mathbf{0}, \boldsymbol{\Sigma}_k),
    \label{eq:channel}
\end{equation}
where the dependence on the port positions appears in the channel spatial correlation matrix $\boldsymbol{\Sigma}_k \in \mathbb{C}^{N \times N}$. In general, the structure of $\boldsymbol{\Sigma}_k$ is determined by the fluid antenna topology, i.e., the relative positions of the antenna ports within the overall aperture. For a 1D fluid antenna, Jakes’ correlation model is widely adopted. However, this model does not generalize properly to planar (2D) arrays with $N_1\times N_2$ ports. For such systems, we will employ Clarke’s model to define $\boldsymbol{\Sigma}_k$,  considering $N=N_1\times N_2$ \cite{RaMoWo24}. In all instances, we assume that the ports in each dimension are equally spaced and denote by $W$, (or equivalently $W_1\times W_2$), the length of the fluid antenna normalized by the wavelength.

To evaluate the performance of the proposed setup, we consider the \ac{SINR} of each user, as follows
\begin{equation}
    \text{SINR}_k = \frac{|\mathbf{w}_k^H\mathbf{S}_k^T\mathbf{H}_k\mathbf{p}_k|^2}{\sum_{j \neq k} |\mathbf{w}_k^H\mathbf{S}_k^T\mathbf{H}_k\mathbf{p}_j|^2 + \frac{1}{\mathrm{SNR}}},
\end{equation}
as it is directly related to the achievable \ac{SE} $R_k$, with $\mathrm{SNR}=\sigma_S^2/\sigma^2$ being the transmit SNR. Thus, we define the problem formulation as follows
\begin{align}
\underset{\{\mathbf S_k\in\mathcal{B},\mathbf w_k\}_{k=1}^K}{\text{max}}
\quad & \hspace{-4mm}\sum\nolimits_{k=1}^K R_k 
\;=\;\sum\nolimits_{k=1}^K \log_2\big(1 + \text{SINR}_k\big). 
\label{eq:probForm}
\end{align}
Since FAMA operates in open-loop, each user can design their corresponding $\mathbf S_k$ and $\mathbf w_k$ independently. From the perspective of the $k$-th user, the problem formulation requires the joint design of the port selection matrix $\mathbf S_k$ and the combining vector $\mathbf w_k$. The FAMA-based literature has primarily focused in the baseline case where $L=1$; hence, the usual strategy is to select the port that maximizes the \ac{SINR}, and a trivial scalar value of one is used as the combining vector.

For the cases of CUMA using more than one RF chain \cite{CUMA}, $\mathbf S_k$ and $\mathbf w_k$ are designed sequentially: first, the selection matrix is configured following a sign-based criterion for selecting a subset of ports for which the desired signal adds up constructively. Then, direct detection or matrix inversion-based strategies are proposed to implement $\mathbf w_k$. Taking a look at the related literature regarding classical antenna selection schemes, algorithms that avoid an exhaustive search over all candidate antenna subsets to design $\mathbf S_k$ are described, e.g. \cite{GoNaPa00,GhGe04,Go02, MeRuAlGoHe15}. In general, these approaches rewrite the performance metric $R_k$ by isolating the contribution of a particular column of the channel matrix. However, such formulations are well-suited for a single-user scenario, but cannot be adapted to the multi-user case in FAMA due to the inter-user interference term.

\section{SINR-based port selection and combining}
\label{sec:SINR}
In this section, we focus on the design of the selection matrix and the combining vector for a FA-based multiport receiver equipped with $L>1$ RF chains.
\subsection{Digital Combining}
We first aim to generalize the port selection criterion in slow-FAMA to the case with $L$ active ports. Similar to \cite{GhGe04,MeRuAlGoHe15}, we solve the problem formulation in \eqref{eq:probForm} in two stages. First, we start by designing the matrix $\mathbf S_k$ following the slow \ac{FAMA} heuristic in \cite{SlowFAMA} (i.e., selecting the $L$ ports with the best SINR). Then, after the active ports are decided, we determine $\mathbf w_k$ aiming to maximize the SINR after combination. This scheme will be referred to as \ac{DC}.

Let us denote as $\mathbf h_{k,r}^T$ the $r$-th row of $\mathbf{H}_k$. Then, for the canonical precoding vectors $\mathbf p_k$ described in Sect. \ref{sec:SM}, we can formulate the slow \ac{FAMA} port selection procedure with $L$ ports as follows
\begin{equation}
     \underset{\mathcal{N}_k \subset \{1, 2, \dots, N\},\, |\mathcal{N}_k| = L}{\operatorname{arg\,max}} \sum_{r \in \mathcal{N}_k} \frac{|\mathbf h_{k,r}^T \mathbf p_k|^2}{\|\mathbf h_{k,r}\|_2^2-|\mathbf h_{k,r}^T \mathbf p_k|^2+ \frac{1}{\mathrm{SNR}}}.
     \label{eq:fama}
\end{equation}
By finding the solution to this formulation, we build the selection matrix with the canonical vectors corresponding to the rows in $\mathcal{N}_k $, i.e. $\mathbf S_k=[\mathbf e_{\mathcal{N}_k(1)},\mathbf e_{\mathcal{N}_k(2)},\ldots,\mathbf e_{\mathcal{N}_k(L)}]$. For the given selection matrix $\mathbf S_k$, we need to design an adequate combining vector $\mathbf w_k$. This vector should increase the signal strength and, at the same time, reduce the inter-user interference. In other words, we aim to solve  
\begin{equation}
    \underset{\mathbf w_k}{\text{max}} \frac{\mathbf{w}_k^H\tilde{\mathbf A}_k\mathbf{w}_k}{\mathbf{w}_k^H\tilde{\mathbf B}_k\mathbf{w}_k},
    \label{eq:Rayquot}
\end{equation}
where we introduced $\tilde{\mathbf A}_k=\mathbf{S}_k^T\mathbf A_k\mathbf{S}_k$ as a submatrix of $\mathbf A_k=\mathbf{H}_k\mathbf{p}_k\mathbf{p}_k^H\mathbf{H}_k^H$,  and $\tilde{\mathbf B}_k=\mathbf{S}_k^T\mathbf B_k\mathbf{S}_k$ as a positive definite submatrix of $\mathbf B_k=\sum_{j \neq k}\mathbf{H}_k\mathbf{p}_j\mathbf{p}_j^H\mathbf{H}_k^H+\mathbf I_R/\mathrm{SNR}$. Due to $\mathbf{S}_k$, the selected row and column indices are common for both $\tilde{\mathbf A}_k$ and $\tilde{\mathbf B}_k$.
Hence, finding the optimal combining vector $\mathbf w_k$ can be recast as a generalized eigenvalue problem \cite{GhKaCr23}, and solved by the dominant generalized eigenvector of the matrix pair $(\tilde{\mathbf A}_k,\tilde{\mathbf B}_k)$ \cite{ScBo04}. In particular, note that defining the vector $\mathbf x_k=\tilde{\mathbf B}_k^{-1/2}\mathbf w_k$ we can find a solution using that $\mathbf x_k$ is the dominant eigenvector of $\tilde{\mathbf C}_k=\tilde{\mathbf B}_k^{-H/2}\tilde{\mathbf A}_k\tilde{\mathbf B}_k^{-1/2}$, and the achievable SINR for the $k$-th user is the associated dominant eigenvalue.

\subsection{Generalized Eigenvector Port Selection}
The extension of the slow-FAMA port selection scheme to the multiport case from in \eqref{eq:fama} aims to select the most promising $L$ ports from the SINR perspective. However, under this criterion the effect of spatial correlation among ports is neglected in the process. In other words, DC selects isolated entries from $\mathbf A$ and $\mathbf B$, but the structure of such matrices is ignored. Besides, the sequential design of the selection matrix and the combining vector may not yield optimal performance. Hence, it is legitimate to ask whether some improvement over the DC scheme can be attained when these aspects are taken into consideration. In the sequel, for the sake of notational simplicity, the dependence on the user index $k$ is dropped.

We propose an alternative strategy that \textit{jointly} designs $\mathbf S$ and $\mathbf w$, and exploits the spatial information provided by $\mathbf A$ and $\mathbf B$. This scheme will be referred to as \ac{GEPort}. Let us take as starting assumption that we can afford to select all $N$ ports at the FA. Then, let us evaluate what is the SINR loss introduced by switching off one of the $N$ available ports. This measure will be obtained by considering the SINRs related to the matrix entries as well as the matrix structures, potentially achieving improved system performance. We define the proposed metric in the following lemma.
\begin{lemma}\label{thm}
Let $\mathbf v_1, \mathbf v_2, \ldots,\mathbf v_N$ be the generalized eigenvectors of the matrix pair $(\mathbf A,\mathbf B)$, associated with the corresponding generalized eigenvalues $\lambda_1 \leq \lambda_2 \dots \leq \lambda_N$, and $v_{N,l}$ the $l$-th element of $\mathbf v_N$. In addition, let $\alpha_{l,1} \leq \alpha_{l,2} \dots \leq \alpha_{l,N-1}$ be generalized eigenvalues of the matrix pair  $(\tilde{\mathbf A}_l,\tilde{\mathbf B}_l)$ obtained by disabling the $l$-th port (i.e. removing row $l$ and column $l$ from $\mathbf A$ and $\mathbf B$). Then, the SINR drop $\delta_{l}$ due to deactivating port $l$ is given by
\begin{equation}
    \delta_{l}=|v_{N,l}|^2(\lambda_N-\lambda_{N-1})\prod_{t=1}^{N-2}\frac{(\lambda_N-\lambda_t)}{(\lambda_N-\alpha_{l,t})}.
    \label{eq:SINRdrop}
\end{equation}
\end{lemma}
\begin{proof}\label{proof-th1}
See Appendix. 
\end{proof}
The result in Lemma \eqref{thm} is useful to the discard the port that presents the lowest performance reduction. As we will later discuss, this process can be iterated until we reach the subset of $L$ ports with the lowest aggregate SINR loss. After observing Lemma \eqref{thm}, some relevant remarks are in order: Note that the first and third factors in \eqref{eq:SINRdrop} depend on the candidate port $l$. Unfortunately, the third factor includes $\alpha_{l,n}$ in the denominator, which is unknown unless the generalized eigenvalues of the reduced matrix pair $(\tilde{\mathbf{A}}, \tilde{\mathbf{B}})$ are computed. Obtaining $\alpha_{l,t}$ $\forall l,t$ can be impractical when the number of ports $N$ grows large. To avoid this and reduce complexity, we can exploit the Cauchy interlacing theorem, which states that the eigenvalues of a Hermitian matrix interlace with those of any of its principal submatrices. Specifically, the eigenvalues of a matrix and a submatrix satisfy \cite{golub2013matrix}
\begin{equation}
    \lambda_{i+1} \geq \alpha_{l,i} \geq \lambda_i,\;\forall i\in\{1,\ldots,N-1\},\;\forall l\in\{1,\ldots,N\}.
\end{equation}
As such, all the differences in the product of the third factor \eqref{eq:SINRdrop} are positive valued, and the quotients for such a factor are smaller than or equal to $1$. Consequently, we set a lower bound for the SINR drop as follows
\begin{equation}
    \delta_{l}\geq|v_{N,l}|^2(\lambda_N-\lambda_{N-1}).
    \label{eq:SINRdropApp}
\end{equation}
Note also that the first factor $|v_{N,l}|^2$, which is given by the entries of the dominant generalized eigenvector of the complete matrix pair $(\mathbf A,\mathbf B)$, depends on $l$, while the remaining factor in \eqref{eq:SINRdropApp} is equal for different ports. Moreover, when $|v_{N,l}|^2\approx 0$, the SINR drop satisfies $\delta_{l}\approx 0$ and deactivating the port $l$ has a negligible impact on the performance. As a consequence, we are interested in the {index $l$ attaining the} smallest value of $|v_{N,l}|^2$, potentially leading to a small performance reduction when the port $l$ becomes inactive.

\begin{algorithm}[!t]
\caption{Generalized Eigenvector Port Selection (GEPort)}
\label{alg}
\begin{algorithmic}[1]
\STATE \textbf{Initialization:} $n \gets 0$, $\mathcal{N}\gets \emptyset$,  $\tilde{\mathbf A}\gets \mathbf A,\tilde{\mathbf B}\gets \mathbf B$ 
\WHILE{$N-n > L$}
    \STATE $\mathbf v_{N-n}\gets$ Power method for $(\tilde{\mathbf A},\tilde{\mathbf B})$
    \STATE $\lambda_{N-n}\gets$ Compute with \eqref{eq:eigenvalue}
    \STATE $l^*\gets\arg \min_l |v_{N-n,l}|^2$
    \STATE $\mathcal{N}\gets \mathcal{N}\cup\{l^*\}$\
    \STATE $\tilde{\mathbf A}\gets $ Remove row and column $l^*$ from $\tilde{\mathbf A}$
    \STATE $\tilde{\mathbf B}\gets $ Remove row and column $l^*$ from $\tilde{\mathbf B}$
    \STATE $\delta\gets \lambda_{N}-\lambda_{N-n}$. Accumulated performance loss 
    \STATE $n\gets n+1$
\ENDWHILE
\STATE \textbf{Return:} $\mathbf S=[\mathbf e_{\mathcal{N}(1)},\mathbf e_{\mathcal{N}(2)},\ldots,\mathbf e_{\mathcal{N}(L)}]$, $\mathbf{w}=\mathbf v_{L}$
\end{algorithmic}
\end{algorithm}

Recall that Lemma \ref{thm} considers the effect of removing the first port. As previously mentioned, the same idea can be repeated until the desired number of ports $L$ is selected. Alternatively, ports can be deactivated to achieve a desired spectral efficiency value or until a predefined performance reduction is attained. In particular, we propose a criterion for building the matrix $\mathbf S$ based on the iterative application of the result in Lemma \ref{thm} for the determination of the set $\mathcal{N}$. Hence, for the $n$-th step in the iteration procedure, we have that
\begin{equation}
    \mathcal{N}=\mathcal{N}\cup\{l^*\}\;\text{s.t.} \;l^*=\arg \min_l |v_{N-n,l}|^2,
\end{equation}
where $\mathbf v_{N-n}$ is the dominant generalized eigenvector for the matrix pair $(\tilde{\mathbf A},\tilde{\mathbf B})$, which are obtained by removing the $n-1$ rows and the columns with indices belonging to the set $\mathcal{N}$, and $v_{N-n,l}$ is the $l$-th entry of such a vector. The procedure proposed to find the dominant generalized eigenvector $\mathbf v_N$ is the so-called power method \cite{golub2013matrix}, which starts with an initial candidate for the eigenvector, $\mathbf t$, and refines the current approximation using $\mathbf t = \tilde{\mathbf B}^{-1} \tilde{\mathbf A}  \mathbf t$. Once $\mathbf t$ converges, the corresponding eigenvalue $\lambda_{N-n}$ is calculated as follows
\begin{equation}
    \lambda_{N-n}=\frac{\mathbf{t}^H\tilde{\mathbf A}\mathbf{t}}{\mathbf{t}^H\tilde{\mathbf B}\mathbf{t}}. \label{eq:eigenvalue}   
\end{equation}
This method is far more efficient than computing the full decomposition. The complete procedure for the GEPort algorithm to jointly select the $L$ ports and design $\mathbf{w}$ is summarized in Alg. \ref{alg}.

\section{Numerical Results}
In this section, we conduct some experiments to assess the performance of the proposed schemes for multiport FAMA. Baseline parameters for the comparison are:  $\sigma_S^2=1$, $M=K=4$ users/antennas at the BS side, a $W=4$ normalized fluid antenna length (1D case) and $W_1\times W_2=4\times 1$ (2D case), $N=100$ FA ports, and $L=2$ active ports. Monte Carlo simulations based on the channel model in \eqref{eq:channel} are used. For benchmarking purposes, we use the cases of conventional slow-FAMA \cite{RaMoWo24} (single RF chain) and \ac{CUMA} (2 RF chains)\cite{CUMA}, and compare them with the DC and GEPort strategies here proposed.


In Fig. \ref{fig:SNR} we evaluate the improvements in terms of the achievable average \ac{SE} of user $k$ for the aforementioned schemes, as a function of the transmit SNR. Solid lines represent the case of users equipped with a 1D FA ($N=100$, $W=4$), while dotted lines correspond to a 2D FA ($N=60\times 15$, $W=4\times 1$). We see that \ac{GEPort} achieves the highest SINR across all SNR levels. More importantly, the performance improvement grows when the system operates in the interference-limited regime (i.e., high transmit SNR). Also, we see that in such a regime, DC behaves similarly as CUMA (even better for the 2D case). Note that while CUMA only operates with 2 RF chains, it requires an arbitrary number of ports to be activated for signal combination.


\begin{figure}[t]
	\centering
	\includegraphics[width=.9\columnwidth]{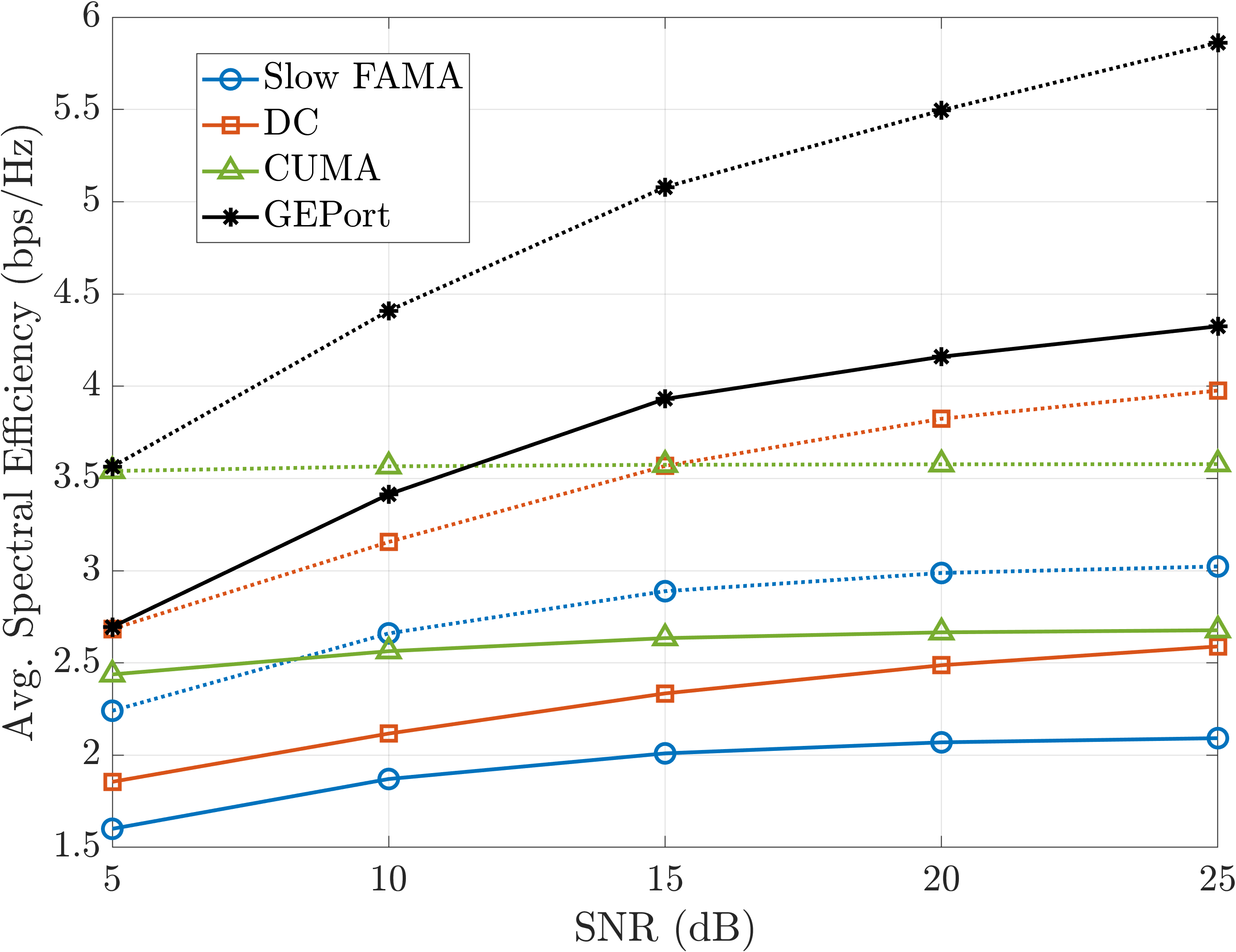}
	\caption{Average spectral efficiency for user $k$ vs. transmit SNR for the different schemes. Solid and dotted lines represent $N=100$ with $W=4$, and $N_1\times N_2=60\times 15$, with $W_1\times W_2=4\times 1$, respectively.} 
 \label{fig:SNR}
\end{figure}

Fig. \ref{fig:L} shows the relationship between the average SE and the number of active ports, $L$. The transmit SNR is fixed and set to $5$dB. As the number of active ports increases (number of RF chains in the proposed schemes), the performance of \ac{DC} and \ac{GEPort} substantially improves. The conventional slow \ac{FAMA} scheme uses single port (i.e., $L=1$); \ac{CUMA} implementation used 2 RF chains (i.e., $L=2$ in our comparison), although it effectively activates a much larger of FA ports as long as the SINR is improved under this combination scheme. Hence, the SE curves associated to these methods are constant. Interestingly, the simpler \ac{DC} combining method performs better than \ac{CUMA} beyond $L=5$. Finally, the performance gap of \ac{GEPort} compared to the reference schemes grows with the number of active ports for the ranges evaluated in our experiment\footnote{Since the number of selected ports in the proposed schemes equals the number of RF chains, a reasonably low value of $L$ should be considered for the sake of complexity reduction in the receiver implementation.}.

\begin{figure}[t]
	\centering
	\includegraphics[width=.9\columnwidth]{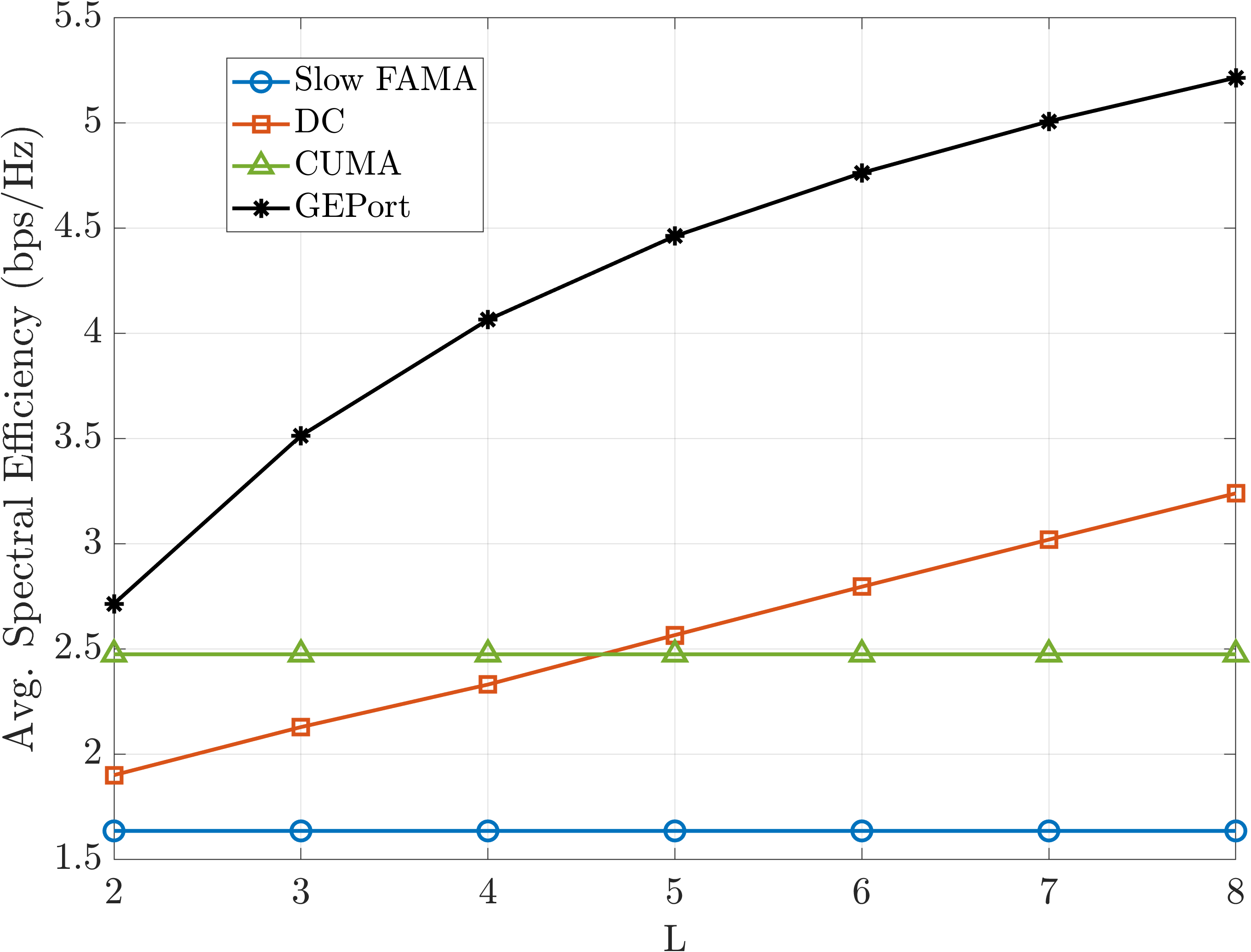}
	\caption{Average spectral efficiency for user $k$ vs. number of active ports $L$. The number of ports is $N=100$ for a size of $W=4$. The transmit SNR is set to $5$dB.} 
 \label{fig:L}
\end{figure}

\begin{figure}[t]
	\centering
	\includegraphics[width=.9\columnwidth]{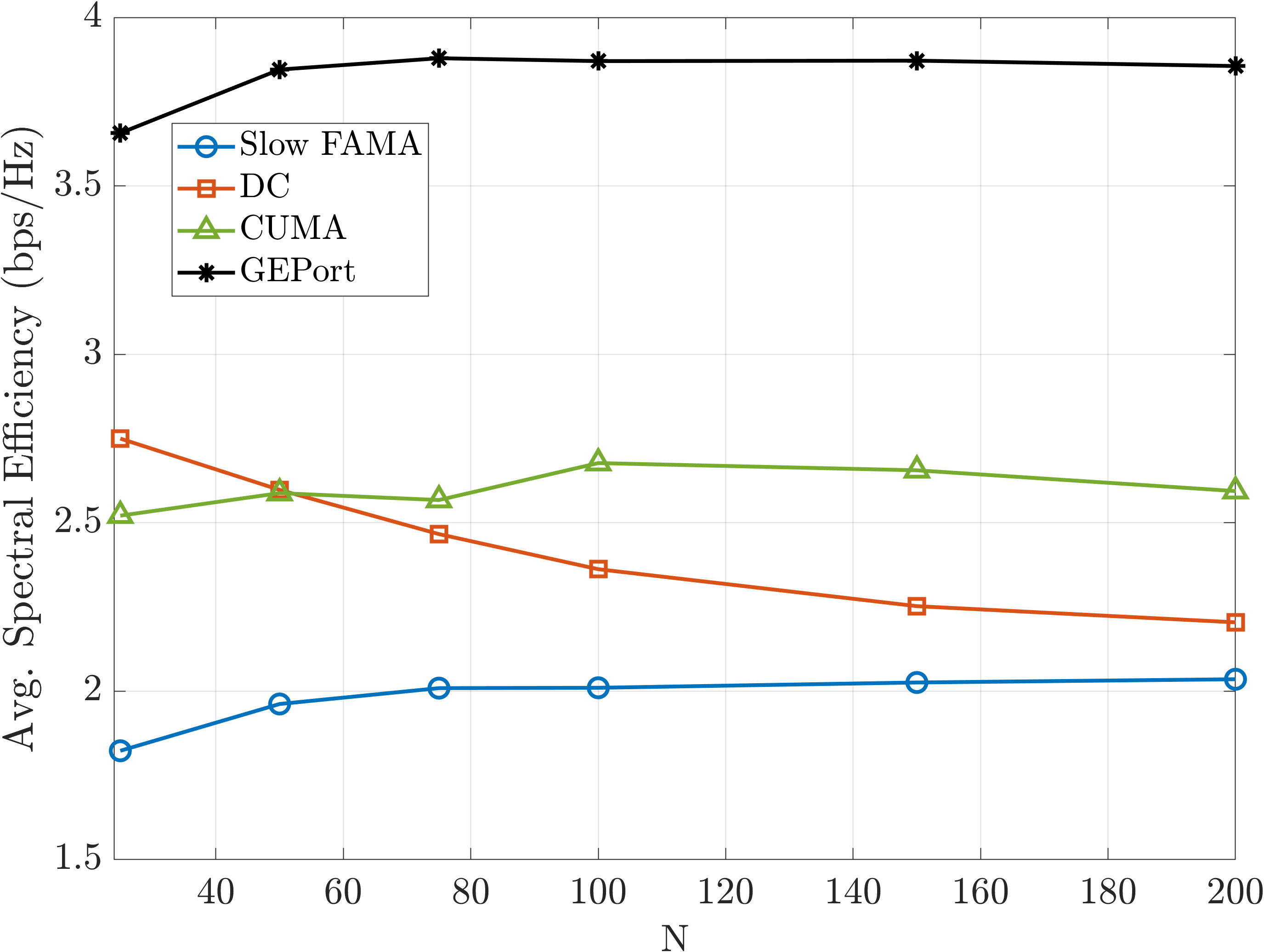}
	\caption{Average spectral efficiency vs. number of available ports $N$. Parameter values are $W=4$, $L=2$ and SNR=$15$dB.} 
 \label{fig:W}
\end{figure}

In Fig. \ref{fig:W} we evaluate the effect of densification in the performance, i.e., increasing $N$ for $W=4$ fixed. The simplest case of $L=2$ has been considered for benchmarking, with a transmit SNR=$15$ dB. We see that the effect of densification (i.e., increasing spatial correlation for the desired and interfering signals) affects the performance of the DC scheme: this confirms the rationale that the sequential design of the selection matrix $\mathbf{S}$ and the combining filter $\mathbf{w}$ (i.e., the multiport extension of slow-FAMA) does not optimally exploit the dependence structures between the matrices $\mathbf{A}$ and $\mathbf{B}$. Conversely, \ac{GEPort} scheme provides a substantial improvement over all competing schemes, although this effect tends to saturate under strong port densification.



\section{Conclusion}
We showcased the potential of multiport reception in FA-assisted multiuser communications. The use of a reduced number of RF chains opens the potential to substantially improve performance under the open-loop slow-FAMA paradigm, as long as the design of the port selection matrix and combining vectors effectively incorporates the correlation structures of signals. This paves the way for the development of new signal processing schemes in FA-enabled multiport receivers.

\section*{Acknowledgment}
The authors gratefully acknowledge Prof. Michael Joham for the insightful discussion that originally inspired the investigation of multiport reception in the context of FAMA, and Dr. Pablo Ram\'irez-Espinosa for his assistance in the slow-FAMA simulations.

\appendix

Consider the Hermitian matrix $\mathbf C\in\mathbb{C}^{N\times N}$, defined as $\mathbf C=\mathbf B^{-H/2}\mathbf A\mathbf B^{-1/2}$, with eigenvalues $\lambda_1\leq \lambda_2\leq \dots\lambda_N$ and corresponding eigenvectors $\mathbf v_1, \dots, \mathbf v_N$.
We denote the $l$-th entry of the eigenvector $\mathbf v_i$ as $v_{i,l}$, and introduce $\tilde{\mathbf A}_l\in\mathbb{C}^{N-1\times N-1}$ and $\tilde{\mathbf B}_l\in\mathbb{C}^{N-1\times N-1}$ as the principal minors of $\mathbf A$, and $\mathbf B$, respectively, obtained by deleting the $l$-th row and the $l$-th column of $\mathbf A$ and $\mathbf B$. Accordingly, $\tilde{\mathbf C}_l=\tilde{\mathbf B}_l^{-H/2}\tilde{\mathbf A}_l\tilde{\mathbf B}_l^{-1/2}$ and we introduce the eigenvalues of $\tilde{\mathbf C}_l$ as $\alpha_{l,1}\leq\alpha_{l,2}\leq\ldots\alpha_{l,N-1}$. Then the following eigenvector-eigenvalue identity holds \cite{Denton_2021}
\begin{equation}
\label{eq:eigenvector_identity}
|v_{i,l}|^2 \prod_{\substack{n=1,n \ne i}}^{N} (\lambda_i - \lambda_n) 
= \prod_{n=1}^{N-1} (\lambda_i - \alpha_{l,n})
\end{equation}
for all $i\in\{1,\ldots,N\}$. Recall that the achievable SINR can be obtained by computing the eigenvalues derived from the Rayleigh quotient in \eqref{eq:Rayquot}. As such, we are interested in finding the indices $l$ and $n$ such that $\alpha_{l,n}$ is maximum. Thus, under the assumption of eigenvalues ordered in ascending order, we get that $n=N-1$ and $l^*=\max_l\alpha_{l,N-1}$ leads to the largest SINR value.
As performing a search over $l$ to find $l^*$ is impractical, we alternatively use the equality in \eqref{eq:eigenvector_identity} to define a performance metric. To that end, we employ the equality for the dominant eigenvalue $\lambda_N$ and the dominant eigenvector $\mathbf v_N$ of $\mathbf C$ with $i=N$, and isolate the term corresponding to the SINR drop, that is 
\begin{equation}
    \lambda_N-\alpha_{l,N-1}=|v_{N,l}|^2\frac{\prod_{\substack{n=1}}^{N-1} (\lambda_N - \lambda_n) }{\prod_{n=1}^{N-2} (\lambda_N - \alpha_{l,n})},
    \label{eq:dropApp}
\end{equation}
getting \eqref{eq:SINRdrop} after rearranging the terms and introducing $\delta_{l}=\lambda_R-\alpha_{l,N-1}$.

\bibliographystyle{IEEEtran}

\bibliography{References}
\end{document}